\documentclass[sunil,ChapterTOCs]{sunil} 
\usepackage{fixltx2e,fix-cm}
\usepackage{amssymb}
\usepackage{amsmath}
\usepackage{graphicx}
\usepackage{subfigure}
\usepackage{makeidx}
\usepackage{multicol}
\usepackage{verbatim}
\usepackage{xcolor}
\usepackage{color,colortbl}
\usepackage{physics}
\usepackage{stmaryrd}
\usepackage{amscd,amsthm}
\usepackage{mathtools}
\usepackage{amsfonts}
\usepackage{url}
\usepackage{etoolbox}
\usepackage{enumitem}

\newcommand{\etal}{\emph{et al.} }
\newcommand{\eg}{\emph{e.g.}}
\newcommand{\ie}{\emph{i.e.}}

\newcommand\inner[1]{\langle #1 \rangle}
\newcommand\dsb[1]{\llbracket #1 \rrbracket}
\def\bv{{\mathbf{v}}}
\def\bw{{\mathbf{w}}}

\def\bu{{\mathbf{u}}}

\def\0{{\mathbf{0}}}
\def\1{{\mathbf{1}}}
\def\ba{{\mathbf{a}}}
\def\bb{{\mathbf{b}}}
\def\bc{{\mathbf{c}}}
\def\Ee{{\rm E}}
\def\EE{{\mathcal{E}}}

\newcommand{\Euc}{\mathop{{\rm E}}}
\newcommand{\HH}{\mathop{{\rm H}}}
\DeclareMathOperator{\wQ}{w_{Q}}
\DeclareMathOperator{\wH}{w_{H}}

\allowdisplaybreaks
\makeindex

\newtheorem{theorem}{Theorem}[section]

\newtheorem{example}{Example}
\newtheorem{definition}{Definition}
\newtheorem{proposition}[theorem]{Proposition}

\begin{document}
\title{Quantum Error-Control Codes} 
\author{Martianus Frederic Ezerman}
\maketitle
\chapter{Quantum Error-Control Codes}

\section{Introduction}\label{intro}

Information is physical. It is sensible to use \textbf{quantum mechanics}\index{quantum mechanics} as a basis of computation and information processing~\cite{Feynman1982}. Here 	at the intersection of information theory, computing, and physics, mathematicians and computer scientists must think in terms of the quantum physical realizations of messages. The often philosophical debates among physicists over the nature and interpretations of quantum mechanics\index{quantum mechanics} shift to harnessing its power for information processing and testing the theory for completeness.

One cannot directly access information stored and processed in massively entangled quantum systems without destroying the content. Turning large-scale \textbf{quantum computing}\index{quantum!computing} into practical reality is massively challenging. To start with, it requires techniques for error control that are much more complex than those implemented effectively in classical systems. As a quantum system\index{quantum system} grows in size and circuit depth, error control becomes ever more important. 

\textbf{Quantum error-control}\index{quantum!error-control} is a set of methods to protect quantum information from unwanted environmental interactions, known as \textbf{decoherence}\index{decoherence}. Classically, one encodes information-carrying vectors into a larger space to allow for sufficient redundancy for error detection and correction. In the quantum setup, information is stored in a subspace embedded in a larger Hilbert space, which is a finite dimensional, normed, vector space over the field of complex numbers $\mathbb{C}$. Codewords are quantum states and errors are operators. 

The good news is that noise, if it can be kept below a certain level, is not an obstacle to resilient quantum computation. This crucial insight is arrived at based on seminal results that form the so-called treshold theorems. Theoretical references include the exposition of Knill \etal in~\cite{Knill1998}, the work of Preskill in~\cite{Preskill1997}, the results shown by Steane in~\cite{Steane2003}, and the paper of Ahoronov and Ben-Or~\cite{Aharonov2008}. A comprehensive review on related experiments is available in~\cite{Campbell2017}. 

The possibility of correcting errors in quantum systems\index{quantum system} was shown, \eg, in the early works of Shor \cite{Shor1995}, Steane \cite{Steane1996} and Laflamme \etal \cite{Laflamme1996}. While the quantum codes\index{quantum code} that these pioneers proposed may nowadays seem to be rather trivial in performance, their construction highlighted both the main obstacles and their respective workarounds. Measurement collapses the information contained in the state into something useless. One should measure the error, not the data. Since repetition is ruled out due to the \textbf{no-cloning theorem}\index{no-cloning theorem}~\cite{Wootters1982}, we use redundancy from spreading the states to avoid repetition. There are multiple types of errors, as we will soon see. The key is to start by correcting the phase errors and, then, use the Hadamard transform to exchange the bit flips and the phase errors. Quantum errors are continuous. Controlling them seemed to be too daunting a task. It turned out that handling a set of discrete \textbf{error operators}\index{quantum error operators|seealso {Pauli matrices}}, represented by \textbf{tensor product}\index{tensor product} of \textbf{Pauli matrices}\index{Pauli matrices}, allows for the control of every $\mathbb{C}$-linear combination of these operators.

Advances continue to be made as effort intensifies to scale quantum computing\index{quantum computing} up. Research in \textbf{quantum error-correcting codes (QECs)}\index{quantum!error-correcting codes}\index{code!quantum| see {quantum, error-correcting codes}} has attracted the sustained attention of established researchers and students alike. Several excellent online lecture notes, surveys, and books are available. Developments up to 2011 have been well-documented in~\cite{Lidar2013}. It is impossible to describe the technical details of every important research direction in QECs. We focus on \textbf{quantum stabilizer codes}\index{quantum code!stabilizer}\index{code!quantum stabilizer} and their variants. The decidedly biased take here is for the audience with more applied mathematics background, including coding theorist, information theorist, researchers in discrete mathematics and finite geometries. No knowledge of quantum mechanics\index{quantum mechanics} is required beyond the very basic. This chapter is meant to serve as an entry point for those who want to understand and get involved in building upon this foundational aspect of quantum information processing and computation, which have been tipped to be indispensable in future technologies. 

A quantum stabilizer\index{quantum code!stabilizer} code is designed so that errors with high probability of occuring transform information-carrying states to an error space which is orthogonal to the original space. The beauty lies in how natural the determination of the location and type of each error in the system is. Correction becomes a simple application of the type of error at the very location.

\section{Preliminaries}\label{sec:prelims}

Consider the field extensions $\mathbb{F}_p$ to $\mathbb{F}_{q=p^r}$ to $\mathbb{F}_{q^m=p^{rm}}$, for positive integers $r$ and $m$. For $\alpha \in \mathbb{F}_{q^m}$, the \textbf{trace mapping}\index{trace map} from $\mathbb{F}_{q^m}$ to $\mathbb{F}_q$ is given by 
\[
\Tr_{\mathbb{F}_{q^m}/\mathbb{F}_q} (\alpha) = \alpha + \alpha^q + \ldots + \alpha^{q^{m-1}} \in \mathbb{F}_q.
\]
The trace of $\alpha$ is the sum of its \textbf{conjugates}\index{field!finite!conjugates in}. If the extension $\mathbb{F}_q$ of $\mathbb{F}_p$ is contextually clear, the notation $\Tr$ is sufficient. Properties of the trace map\index{trace map} can be found in standard textbooks, \eg, \cite[Chapter~2]{Lidl1996}. The key idea is that $\Tr_{\mathbb{F}_{q^m}/\mathbb{F}_q}$ serves as a \emph{description} for all linear transformations from $\mathbb{F}_{q^m}$ to $\mathbb{F}_q$.

Let $G$ be a \textbf{finite abelian group}\index{group!finite abelian}, written multiplicatively, with the identity $1_{G}$. Let $U$ be the multiplicative group\index{group} of complex numbers of modulus $1$, \ie, the unit circle of radius $1$ on the complex plane $\mathbb{C}$. A \textbf{character}\index{character} $\chi : G \mapsto U$ is a homomorphism. For any $g \in G$, the images of $\chi$ are $\abs{G}$-th roots of unity since $\left(\chi(g)\right)^{\abs{G}} = \chi \left(g^{\abs{G}}\right)=1$. Let $\overline{c}$ denote the complex conjugate of $c$. Then $\chi(g^{-1})= \left(\chi(g)\right)^{-1} = \overline{\chi(g)}$. The only trivial character\index{character} is $\chi_{0}: g \mapsto 1$ for all $g \in G$. One can associate $\chi$ and $\overline{\chi}$ by using $\overline{\chi}(g)=\overline{\chi(g)}$. The set of all characters\index{character} of $G$ forms, under composition $\circ$, an abelian group\index{group!abelian} $\widehat{G}$.

For $g, h \in G$ and $\chi, \Psi \in \widehat{G}$ we have two \textbf{orthogonality relations}\index{orthogonality relation}  
\begin{equation}\label{eq:ortho}
	\sum_{g \in G} \chi(g) \, \overline{\Psi(g)} =
	\begin{cases}
		0 \mbox{, if } \chi \neq \Psi\\
		\abs{G} \mbox{, if } \chi = \Psi
	\end{cases}
	\mbox{and} \quad 
	\sum_{\chi \in \widehat{G}} \chi(g) \, \chi(h^{-1}) =
	\begin{cases}
		0 \mbox{, if } g \neq h\\
		\abs{G} \mbox{, if } g = h
	\end{cases}.
\end{equation}

The additive character\index{character} 
$\chi_1 := c \mapsto e^{\frac{2 \pi }{p} \Tr(c)}$, for all $c \in \mathbb{F}_q$, is called the 
\textbf{canonical character}\index{character!canonical}. For a chosen $b \in \mathbb{F}_q$ and for all $c \in \mathbb{F}_q$, 
\[
\chi_b := \mathbb{F}_q \rightarrow U \mbox{ sending } 
c \mapsto \chi_1 (b \cdot c) = e^{\frac{2 \pi}{p} \Tr(b \cdot c)}
\]
is a character of $(\mathbb{F}_q,+)$. Every character\index{character} of $(\mathbb{F}_q,+)$ can, in fact, be expressed in this manner. The extension to $(\mathbb{F}_q^n,+)$ is straightforward. 
\begin{theorem}
	Let $\zeta := e^{\frac{2 \pi}{p}}$ and $\Tr$ be the trace map\index{trace map} with $q=p^m$. Let $\ba = (a_1,a_2,\ldots,a_n)$ and $\bb=(b_1,b_2,\ldots,b_n)$ be vectors in $\mathbb{F}_q^n$. For each $\ba$, 
	\[
	\lambda_{\ba} : \mathbb{F}_q^n \mapsto \{1,\zeta,\zeta^2,\ldots,\zeta^{p-1}\} \mbox{, sending }
	\bb \mapsto \zeta^{\Tr(\ba \cdot \bb)} = \zeta^{\Tr(a_1 b_1+\ldots+a_n b_n)},
	\]
	for all $\bb \in \mathbb{F}_q$, is a character\index{character} of $(\mathbb{F}_q^n,+)$. Hence, 
	$\widehat{\mathbb{F}_q^n} = \{ \lambda_{\ba} : \ba \in \mathbb{F}_q^n \}$.
\end{theorem}

A \textbf{qubit}\index{qubit}\index{quantum bit|see {qubit}}, a term coined by Schumacher in~\cite{Schumacher1995}, is the canonical quantum system\index{quantum system} consisting of two distinct levels. The states of a qubit live in $\mathbb{C}^2$ and are defined by their continuous amplitudes. A \textbf{qudit}\index{qudit} refers to a system of $q \geq 3$ distinct levels, with a \textbf{qutrit}\index{qutrit} commonly used when $q=3$. Physicists prefer the ``\textbf{bra}\index{bra}'' $\ket{\cdot}$ and ``\textbf{ket}\index{ket}'' $\bra{\cdot}$ notation to describe quantum systems. A $\ket{\varphi}$ is a (column) vector while $\bra{\psi}$ is the vector dual of $\ket{\psi}$.

\begin{definition}[\textbf{Quantum systems}\index{quantum system}]
	A qubit\index{qubit} is a nonzero vector in $\mathbb{C}^{2}$, usually with basis $\{\ket{0},\ket{1}\}$. It is written in vector form as $\ket{\varphi}:= \alpha \ket{0} + \beta \ket{1}$, or in matrix form as 
	$\displaystyle{\begin{bmatrix} \alpha \\ \beta\end{bmatrix}}$, with 
	$\norm{\alpha}^2 + \norm{\beta}^2 = 1$ .
	
	An $n$-qubit system or vector is a nonzero element in $\left(\mathbb{C}^2\right)^{\otimes n} \cong \mathbb{C}^{2^n}$. Let $\ba=(a_1,\ldots,a_{n}) \in \mathbb{F}_2^n$. The standard $\mathbb{C}$-basis is 
	\[
	\{\ket{a_1 a_2 \ldots a_{n}} :=
	\ket{a_1} \otimes \ket{a_2} \otimes \ldots \otimes \ket{a_{n}}: \ba \in \mathbb{F}_2^n\}. 
	\]
	An arbitrary nonzero vector in $\mathbb{C}^{2^n}$ is written
	\[
	\ket{\psi} = \sum_{\ba \in \mathbb{F}_2^n} c_{\ba} \ket{\ba} \mbox{, with } c_{\ba} \in \mathbb{C} \mbox{ and } \frac{1}{2^n} \sum_{\ba \in \mathbb{F}_2^n} \norm{c_{\ba}}^2 = 1.
	\]
	The normalization is optional since $\ket{\psi}$ and $\alpha \ket{\psi}$ are considered the same state for nonzero $\alpha \in \mathbb{C}$.
\end{definition}	

The \textbf{inner product} of $\ket{\psi} : = \sum_{\ba \in \mathbb{F}_2^n} c_{\ba} \ket{\ba}$ and 
$\ket{\varphi} : = \sum_{\ba \in \mathbb{F}_2^n} b_{\ba} \ket{\ba}$ is 
\[
\braket{\psi}{\varphi}= \sum_{\ba \in \mathbb{F}_2^n} \overline{c_{\ba}} \, b_{\ba}.
\]
Their (Kronecker) \textbf{tensor product}\index{tensor product} is written as $\ket{\varphi} \otimes \ket{\psi}$ and is often abbreviated to $\ket{\varphi \psi}$. The states $\ket{\psi}$ and $\ket{\varphi}$ are \textbf{orthogonal}\index{orthogonal state} or \textbf{distinguishable} if $\braket{\psi}{\varphi}=0$. Let $A$ be a $2^n \times 2^n$ complex unitary matrix with conjugate transpose $A^{\dagger}$. The (Hermitian) inner product of $\ket{\varphi}$ and $A \ket{\psi}$ is equal to that of $A^{\dagger} \ket{\varphi}$ and $\ket{\psi}$. Henceforth, $i := \sqrt{-1}$.

\begin{definition}[\textbf{Qubit error operators}]
	A qubit {\it error operator}\index{qubit}\index{quantum error operators|seealso {Pauli matrices}} is a unitary $\mathbb{C}$-linear operator acting on $\mathbb{C}^{2^n}$ qubit by qubit. It can be expressed by a unitary matrix with respect to the basis $\{ \ket{0},\ket{1}\}$. The three {\it nontrivial} errors acting on a qubit are known as the \textbf{Pauli matrices}\index{Pauli matrices}:
	\begin{equation}\label{eq:Pauli}
		\sigma_x =
		\begin{bmatrix}
			0 & 1 \\ 1 & 0
		\end{bmatrix}, 
		\quad 
		\sigma_z = 
		\begin{bmatrix}
			1 & 0 \\ 0 & -1
		\end{bmatrix}, 
		\quad 
		\sigma_y = i \, \sigma_x \, \sigma_z =
		\begin{bmatrix}
			0 & -i \\ i & 0
		\end{bmatrix}.
	\end{equation}
\end{definition}	

The actions of the error operators\index{quantum error operators|seealso {Pauli matrices}} on a qubit\index{qubit} $\ket{v} = \alpha \ket{0} + \beta \ket{1} \in \mathbb{C}^2$ can be considered based on their types. The \textbf{trivial operator} $I_2$ leaves the qubit unchanged. The \textbf{bit-flip error}\index{quantum error!bit-flip} $\sigma_x$ flips the probabilities
\[
\sigma_x \, \ket{\varphi}=\beta \ket{0} + \alpha \ket{1} \mbox{ or }
\sigma_x 
\begin{bmatrix} \alpha \\ \beta
\end{bmatrix}
=
\begin{bmatrix} \beta \\ \alpha
\end{bmatrix}.
\]
The \textbf{phase-flip error}\index{quantum error!phase-flip} $\sigma_z$ modifies the angular measures 
\[
\sigma_z \, \ket{\varphi}=\alpha \ket{0} - \beta \ket{1} \mbox{ or }
\sigma_z 
\begin{bmatrix} \alpha \\ \beta
\end{bmatrix}
=
\begin{bmatrix} \alpha \\ - \beta
\end{bmatrix}.
\]
The \textbf{combination error}\index{quantum error!combination} $\sigma_y$ contains both bit-flip\index{quantum error!bit-flip} and phase-flip\index{quantum error!phase-flip}, implying 
\[
\sigma_y \, \ket{\varphi}=-i \beta \ket{0} + i \alpha \ket{1} \mbox{ or }
\sigma_y 
\begin{bmatrix} \alpha \\ \beta
\end{bmatrix}
=
\begin{bmatrix} -i \beta \\ i \alpha
\end{bmatrix}.
\]
It is immediate to confirm that $\sigma_x^2 = \sigma_y^2 = \sigma_z^2 = I_2$ and $\sigma_x \sigma_z = - \sigma_z \sigma_x$. The Pauli matrices\index{Pauli matrices} generate a group\index{group} of order $16$. Each of its elements can be uniquely represented as $i^{\lambda} w$, with $\lambda \in \{0,1,2,3\}$ and $w \in \{I_2,\sigma_x, \sigma_z, \sigma_y\}$. 

\section{The Stabilizer Formalism}\label{sec:stab}

The most common route from classical coding theory to QEC is via the \textbf{stabilizer formalism}\index{quantum stabilizer formalism}, from which numerous specific constructions emerge. Classical codes can \emph{not} be used \emph{as} quantum codes\index{quantum code} but can \emph{model} the error operators\index{quantum error operators|seealso {Pauli matrices}} in some quantum channels. The capabilities of a QEC can then be inferred from the properties of the corresponding classical codes. The main tools come from \textbf{character theory}\index{character} and \textbf{symplectic geometry}\index{symplectic!geometry} over \textbf{finite fields}\index{fields!finite}. Our focus is on the qubit\index{qubit} setup since it is the most deployment-feasible and because the general qudit\index{qudit} case naturally follows from it.

Let $\ba =(a_1,a_2,\ldots,a_n) \in \mathbb{F}_2^n$, $\lambda \in \{0,1,2,3\}$, and $w_j \in \{I_2,\sigma_x,\sigma_z,\sigma_y\}$. A \textbf{quantum error operator}\index{quantum error operators|seealso {Pauli matrices}} on $\mathbb{C}^{2^n}$ is of the form $\Ee := i^{\lambda} w_1 \otimes w_2 \otimes \ldots \otimes w_n$. It is a $\mathbb{C}$-linear unitary operator acting on a $\mathbb{C}^{2^n}$-basis $\{ \ket{\ba} = \ket{a_1} \otimes \ket{a_2} \otimes \ldots \otimes \ket{a_n} \}$ by 
$\Ee \ket{\ba} : = i^{\lambda} \left(w_1 \ket{a_1} \otimes w_2 \ket{a_2} \otimes \ldots \otimes w_n \ket{a_n}\right)$. The \textbf{set of error operators}\index{quantum error operators|seealso {Pauli matrices}}
\[
\EE_n := \{ i^{\lambda} w_1 \otimes w_2 \otimes \ldots \otimes w_n \}
\]
is a non-abelian group\index{group!non-abelian} of cardinality $4^{n+1}$. Given $\Ee := i^{\lambda} w_1 \otimes w_2 \otimes \ldots \otimes w_n$ and $\Ee' := i^{\lambda'} w_1' \otimes w_2' \otimes \ldots \otimes w_n'$ in $\EE_n$, we have 
\begin{align*}
	\Ee \, \Ee' &= 
	i^{\lambda + \lambda'} (w_1 w_1') \otimes (w_2 w_2') \otimes \ldots \otimes (w_n w_n')\\
	&= i^{\lambda + \lambda' + \lambda''} w_1'' \otimes w_2'' \otimes \ldots \otimes w_n'' 
	\mbox{, where } w_j w_j' = i^{\lambda_j''} w_j'' \mbox{ and } \lambda'' = \sum_{j=1}^n \lambda_j''.
\end{align*}
Expanding $\Ee'\, \Ee$ makes it clear that $ \Ee \, \Ee' = \pm 1 ~\Ee'\, \Ee$.
\begin{example}\label{ex1}
	Given $n=2$, $\Ee = I_2 \otimes \sigma_x$ and $\Ee'=\sigma_z \otimes \sigma_y$, we have 
	$\Ee \, \Ee' = \sigma_z \otimes \sigma_x \sigma_y = \sigma_z \otimes i \sigma_z = i \sigma_z \otimes  \sigma_z $ and  
	$\Ee' \, \Ee = \sigma_z \otimes \sigma_y \sigma_x = \sigma_z \otimes i^{3} \sigma_z = i^{3} \sigma_z \otimes \sigma_z$. \quad \qedsymbol
\end{example}

The \textbf{center}\index{group!center} of $\EE_n$ is $\mathcal{C}(\EE_n):=\{i^{\lambda} I_2 \otimes I_2 \otimes \ldots \otimes I_2\}$. Let $\overline{\EE_n}$ denote the \textbf{quotient group}\index{group!quotient} $\EE_n / \mathcal{C}(\EE_n)$ of cardinality $\abs{\overline{\EE_n}}=4^n$. This group is an abelian $2$-elementary\index{group!abelian $2$-elementary} group $\cong (\mathbb{F}_2^{2n},+)$, since $\overline{\Ee}^2 = I_2 \otimes \cdots \otimes I_2 = I_{2^n}$ for any $\overline{\Ee} \in \overline{\EE_n}$.

We switch notation to define the product of error operators\index{quantum error operators|seealso {Pauli matrices}} in terms of an inner product of their vector representatives. We write $\Ee = i^{\lambda} w_1 \otimes w_2 \otimes \ldots \otimes w_n$ as $\Ee = i^{\lambda + \epsilon} X(\ba) Z(\bb)$, where $\ba=(a_1,\ldots,a_n), \bb=(b_1,\ldots,b_n) \in \mathbb{F}_2^n$ and $\epsilon := \abs{ \{ 1 \leq i \leq n\} : w_i = \sigma_y }$, by replacing $(a_i,b_i)$ with $(0,0)$ if $w_i = I_2$, by $(1,0)$ if $w_i = \sigma_x$, by $(0,1)$ if $w_i = \sigma_z$, and by $(1,1)$ if $w_i = \sigma_y$. 

The respective actions of $X(\ba)$ and $Z(\bb)$ on any vector $\ket{\bv} \in \mathbb{C}^{2^n}$, for $\bv \in \mathbb{F}_2^n$, are $X(\ba) \ket{\bv} = \ket{\ba + \bv}$ and 
$Z(\bb) \ket{\bv} = (-1)^{\bb \cdot \bv} \ket{\bv}$. The matrix for $X(\ba)$ is a symmetric $\{0,1\}$ matrix. It represents a permutation consisting of $2^{n-1}$ transpositions. The matrix for $Z(\bb)$ is diagonal with diagonal entries $\pm 1$. Hence, writing the operators in $\EE_n$ as $\Ee := i^{\lambda} X(\ba) Z(\bb)$ and $\Ee' := i^{\lambda'} X(\ba') Z(\bb')$, one gets $\Ee \, \Ee' = (-1)^{\ba \cdot \bb' + \ba' \cdot \bb} ~\Ee' \, \Ee$. The \textbf{symplectic inner product}\index{symplectic!inner product|seealso{inner product, symplectic}}\index{inner product!symplectic} of $(\ba|\bb)$ and $(\ba'|\bb')$ in $\mathbb{F}_2^{2n}$ is
\begin{equation}\label{eq:symplectic}
	\inner{(\ba|\bb), (\ba'|\bb')}_s = \ba \cdot \bb' + \ba' \cdot \bb 
\end{equation}
or, in matrix form,
\[
\inner{(\ba|\bb), (\ba'|\bb')}_s = 
\begin{bmatrix}
	\ba & \bb 
\end{bmatrix}
\,
\begin{bmatrix}
	0 & I_n\\
	I_n & 0 
\end{bmatrix}
\,
\begin{bmatrix}
	\ba' \\ \bb' 
\end{bmatrix}.
\]
The \textbf{symplectic dual}\index{symplectic!dual} of $C \subseteq \mathbb{F}_2^{2n}$ is 
$C^{\perp_s} = \{ \bu \in \mathbb{F}_2^{2n} : \inner{\bu,\bc}_s=0 \, \forall \, \bc \in C \}$. 
Thus, a subgroup $G$ of $\EE_n$ is abelian if and only if $\overline{G}$ is a symplectic self-orthogonal\index{symplectic!self-orthogonal} subspace of $\overline{\EE_n} \cong \mathbb{F}_2^{2n}$.

\begin{example}\label{ex2}
	Continuing from Example~\ref{ex1}, we write $\Ee=X((0,1)) Z((0,0))$ and $\Ee'=i X((0,1)) Z((1,1))$. We choose the ordering $(0,0),(0,1),(1,0),(1,1)$ of $\mathbb{F}_2^2$ and the corresponding ordering for the basis of $\mathbb{C}^4$. The matrix for $X((0,1))$ agrees with $I_2 \otimes \sigma_x$, the matrix for $Z((0,0))$ is $I_4$, and the matrix for $Z((1,1))$ is diagonal with diagonal entries $1,-1,-1,1$. Multiplying matrices confirms that $\sigma_z \otimes \sigma_y$ is indeed $i X((0,1)) Z((1,1))$. \quad \qedsymbol
\end{example}
The \textbf{quantum weight}\index{quantum weight} of an error operator\index{quantum error operators|seealso {Pauli matrices}} $\Ee = i^{\lambda} X(\ba) Z(\bb) \in \EE_n$ is
\begin{align*}
	\wQ(\Ee) := \wQ(\overline{\Ee}) &= \wQ(\ba | \bb)= \abs{ \{ 1 \leq i \leq n : a_i = 1 \mbox{ or } b_i=1\} } \\
	& = \abs{ \{ 1 \leq i \leq n : w_i \neq I_2 \}}.
\end{align*}
By definition, $\wQ(\Ee \, \Ee') \leq \wQ(\Ee) + \wQ(\Ee')$, for any $\Ee, \Ee' \in \EE_n$. We can define the \textbf{set of all error operators of weight at most $\delta$}\index{quantum error operators|seealso {Pauli matrices}} in $\EE_n$ and determine its cardinality. Let
\[
\EE_n(\delta) := \{ \Ee \in \EE_n : \wQ(\Ee) \leq \delta \} \mbox{ and }
\overline{\EE_n}(\delta) = \{ \overline{\Ee} \in \overline{\EE_n} : \wQ(\overline{\Ee}) \leq \delta \}.
\]
Then $\abs{\EE_n(\delta)}=4 \sum_{j=0}^\delta 3^j \binom{n}{j}$ and 
$\abs{\overline{\EE_n}(\delta)}= \sum_{j=0}^\delta 3^j \binom{n}{j}$. 

In the classical setup, both errors and codewords are vectors over the same field. In the quantum setup, errors are linear combinations of the tensor products\index{tensor product} of Pauli matrices\index{Pauli matrices}. A qubit\index{qubit} code $Q \subseteq \mathbb{C}^{2^n}$ has three parameters: its \textbf{length} $n$, \textbf{dimension} $K$ over $\mathbb{C}$, and \textbf{minimum distance}\index{quantum code!minimum distance of} $d=d(Q)$. We use 
\[
((n,K,d)) \mbox{ or } \dsb{n,k,d} \mbox{ with } k = \log_2 K
\] 
to signify that $Q$ describes the encoding of $k$ logical qubits as $n$ physical qubits, with $d$ being the smallest number of simultaneous errors that can transform a valid codeword into another.  

\begin{definition}[\textbf{Knill-Laflamme condition}~\cite{Knill1997}\index{Knill-Laflamme condition}] 
	A quantum code\index{quantum code} $Q$ can correct up to $\ell$ quantum errors if the followings hold. If $\ket{\varphi}, \ket{\psi} \in Q$ are distinguishable, \ie, $\braket{\varphi}{\psi}=0$, then $\bra{\varphi} \Ee_1 \Ee_2 \ket{\psi} = 0$, \ie, $\Ee_1 \ket{\varphi}$ and $\Ee_2 \ket{\psi}$ must remain distinguishable, for all $\Ee_1,\Ee_2 \in \EE_n(\ell)$. The minimum distance\index{quantum code!minimum distance of} of $Q$ is $d:=d(Q)$ if $\bra{\varphi} \Ee \ket{\psi}=0$ for all $\Ee \in \EE_n(d-1)$ and for all distinguishable $\ket{\varphi}, \ket{\psi} \in Q$.
\end{definition}

Given an $((n,K,d))$-qubit code $Q$ and an $\Ee \in \EE_n$, $\Ee \, Q$ is a subspace of $\mathbb{C}^{2^n}$. The fact that $Q$ corrects errors of weight up to $\ell = \lfloor \frac{d-1}{2} \rfloor$ does \emph{not} imply that the subspaces $\{\overline{\Ee} Q: \overline{\Ee} \in \overline{\EE_n}(\ell)\}$ are orthogonal to each other. It is possible that a codeword $\ket{\bv}$ is fixed by some $\Ee \neq I_{2^n}$, say, when $\ket{\bv}$ is an \textbf{eigenvector}\index{eigenvector} of $\Ee$ satisfying $\Ee \ket{\bv} = \alpha \ket{\bv}$ for some nonzero $\alpha \in \mathbb{C}$. If the subspaces $\{\overline{\Ee} \, Q: \overline{\Ee} \in \overline{\EE_n}(\ell)\}$ are orthogonal to each other, then $Q$ is said to be \textbf{pure}\index{quantum code!pure}. Otherwise, the code is \textbf{degenerate} or \textbf{impure}\index{quantum code!impure}.

To formally define a \textbf{qubit stabilizer code}\index{qubit}\index{quantum code!stabilizer}, we choose an abelian group\index{group!abelian} $G$, which is a subgroup of $\EE_n$, and associate $G$ with a classical code $C \subset \mathbb{F}_2^{2n}$, which is \textbf{self-orthogonal}\index{symplectic!self-orthogonal} under the symplectic inner product\index{symplectic!inner product|seealso{inner product, symplectic}}\index{inner product!symplectic}. The action of $G$ partitions $\mathbb{C}^{2^n}$ into a direct sum of $\chi$-eigenspaces\index{eigenspace} $Q(\chi)$ with $\chi \in \widehat{G}$. The properties of $Q := Q(\chi)$ follow from the properties of $C$ and $C^{\perp_s}$. The stabilizer formalism\index{quantum stabilizer formalism}, first introduced by Gottesman in his thesis~\cite{Gottesman97} and described in the language of group algebra by Calderbank \etal in~\cite{Calderbank1998}, remains the most widely-studied approach to control quantum errors. Ketkar \etal generalized the formalism to qudit\index{qudit} codes derived from classical codes over $\mathbb{F}_{q^2}$ in~\cite{Ketkar2006}.

Let $G$ be a finite abelian group\index{group!abelian} acting on a finite dimensional $\mathbb{C}$-vector space $V$. Each $g \in G$ is a Hermitian operator of $V$ and, for any $g,g' \in G$ and for all $\ket{\bv} \in V$, 
$(g g') \ket{\bv} = g (g'(\ket{\bv})) \mbox{ and } g g^{-1} (\ket{\bv}) = \ket{\bv}$. Let $\widehat{G}$ be the character\index{character} group\index{group} of $G$. For any $\chi \in \widehat{G}$, the map 
$L_{\chi} := \frac{1}{\abs{G}} \sum_{g \in G} \overline{\chi}(g) \, g$ 
is a linear operator over $V$. The set $\{L_{\chi}: \chi \in \widehat{G}\}$ is the \textbf{system of orthogonal primitive idempotent operators}. 

\begin{proposition}
	$L_{\chi}$ is idempotent, \ie, $L_{\chi}^2=L_{\chi}$ and $L_{\chi} L_{\chi'}=0$ if $\chi \neq \chi'$. The operators in the system sum to the identity $\sum_{\chi \in \widehat{G}} L_{\chi} = 1$. For all $g \in G$, we have $g \, L_{\chi} = \chi(g) \, L_{\chi}$.
\end{proposition}

\begin{proof} In $G$, let $g h =a$, \ie, $h=a g^{-1}$. Using $\overline{\chi}=\chi^{-1}$ and the orthogonality of characters\index{character}\index{orthogonality relations}, we write
	\begin{multline}\label{eq:idem}
		L_{\chi} L_{\chi'} 
		= \frac{1}{\abs{G}^2} \sum_{g \in G} 
		\overline{\chi} (g) \, g ~ \sum_{h \in G} \overline{\chi'} (h) \, h 
		= \frac{1}{\abs{G}^2} \sum_{a,g \in G} \overline{\chi}(g) ~
		\overline{\chi'}(a g^{-1})~ a \\
		= \frac{1}{\abs{G}^2} \sum_{a \in G} \overline{\chi'}(a) \,a ~ 
		\sum_{g \in G} (\overline{\chi} \, \chi') (g).
	\end{multline}
	The third equality comes from collecting terms that contain only $a$ and only $g$. By the first orthogonality relation\index{orthogonality relation} in Equation (\ref{eq:ortho}), one arrives at
	\[
	L_{\chi} L_{\chi'} = \frac{1}{\abs{G}} 
	\sum_{a \in G} \overline{\chi'}(a) a =
	\begin{cases}
		0 \mbox{, if } \chi \neq \chi',\\
		L_{\chi} \mbox{, by definition.}
	\end{cases}
	\]
	We verify the second assertion by using the second orthogonality relation\index{orthogonality relation} in Equation (\ref{eq:ortho}). Since $\overline{\chi}(1)=1$, we obtain
	\[
	\sum_{\chi \in \widehat{G}} L_{\chi} 
	= \frac{1}{\abs{G}} \sum_{\chi \in \widehat{G}} \sum_{g \in G} \overline{\chi}(g) \, g 
	= \frac{1}{\abs{G}} \sum_{g \in G} g \sum_{\chi \in \widehat{G}} \overline{\chi}(g) \, \overline{\chi}(1)= 1.
	\]
	Using $g \, h = a$, the definition of $L_{\chi}$, and the equality $\overline{\chi} (g^{-1}) = \chi(g)$, one gets
	\begin{align*}
		g \, L_{\chi} &= \frac{1}{\abs{G}} \sum_{h \in G} \overline{\chi} (h)\, g \, h 
		= \frac{1}{\abs{G}} \sum_{a \in G} \overline{\chi} (a \cdot g^{-1}) \,a \\
		&= \frac{1}{\abs{G}} \overline{\chi} (g^{-1})  \sum_{a \in G} \overline{\chi} (a) \, a = \chi(g) \, L_{\chi}.
	\end{align*}
\end{proof}

\begin{proposition}
	For each $\chi \in \widehat{G}$, let 
	$V(\chi) := L_{\chi} V = \{L_{\chi}(\ket{\bv}) : \ket{\bv} \in V \}$. For $\ket{\bv} \in V(\chi)$ and $g \in G$, we have $g \ket{\bv} = \chi(g) \, \ket{\bv}$. 
	Thus, $V(\chi)$ is a common eigenspace\index{eigenspace} of all operators in $G$. 
	A direct decomposition $V = \bigoplus_{\chi \in \widehat{G}} V(\chi)$ 
	ensures that each $\ket{\bv} \in V$ has a unique expression 
	\[
	\ket{\bv} = \sum_{\chi \in \widehat{G}} \ket{\bv}_{\chi} \mbox{, where } 
	\ket{\bv}_{\chi} \in V(\chi).
	\]
\end{proposition}
\begin{proof}
	For $\ket{\bv} \in V(\chi)$ and $g \in G$, there exists $\ket{\bw} \in V$ such that 
	\[
	g \ket{\bv} = g \, L_{\chi} (\ket{\bw}) = \chi(g) \, L_{\chi} (\ket{\bw}) = \chi(g) \ket{\bv},
	\]
	confirming the first assertion.
	
	For each $\ket{\bv} \in V$, we write 
	$\ket{\bv} = \left(\sum_{\chi \in \widehat{G}} L_{\chi}\right) \ket{\bv} 
	=\sum_{\chi \in \widehat{G}} \ket{\bv}_{\chi}$, where 
	$\ket{\bv}_{\chi} := L_{\chi}(\ket{\bv}) \in V(\chi)$. On the other hand, if $\ket{\bv} = \sum_{\chi \in \widehat{G}} \ket{\bu_{\chi}}$ for $\ket{\bu_{\chi}} \in V(\chi)$, then $\ket{\bu_{\chi}} = L_{\chi} \ket{\bw_{\chi}}$ for some $\ket{\bw_{\chi}} \in V$. Since $\{L_{\chi} : \chi \in \widehat{G}\}$ has the orthogonality property, for every $\chi \in \widehat{G}$, we have
	\begin{align*}
		\ket{\bv}_{\chi} &= L{\chi} ~\ket{\bv} = L_{\chi} \left(\sum_{\chi' \in \widehat{G}} \ket{\bu_{\chi'}}\right) = L_{\chi} \left(\sum_{\chi' \in \widehat{G}} L_{\chi'} \ket{\bw_{\chi'}}\right)\\
		&=
		\sum_{\chi' \in \widehat{G}} \left(L_{\chi} L_{\chi'} \ket{\bw_{\chi'}}\right) 
		= L_{\chi} \ket{\bw_{\chi}} = \ket{\bu_{\chi}}.
	\end{align*}
	Thus, $V = \bigoplus_{\chi \in \widehat{G}} V(\chi)$. 
\end{proof}

It is also a well-known fact that $V(\chi)$ and $V(\chi')$ are Hermitian orthogonal for all $\chi \neq \chi' \in \widehat{G}$ when all $g \in G$ are unitary linear operators on $V$.

All the tools to connect qubit\index{qubit} stabilizer codes\index{quantum code!stabilizer} to classical codes are now in place. We choose
$G:=\left\langle g_1,g_2,\ldots,g_k\right\rangle$ to be an abelian subgroup of $\EE_n$ with 
$g_j := i^{\lambda_j} ~X(\ba_j)~ Z(\bb_j)$ for $1 \leq j \leq k $, where $\ba_j, \bb_j \in \mathbb{F}_2^n$ and $\lambda_j \equiv \ba_j \cdot \bb_j \pmod{2}$. Since $\sigma_x$ and $\sigma_z$ are \textbf{Hermitian unitary matrices}\index{Hermitian unitary matrices}, $X(\ba_j)$ and $Z(\bb_j)$ are also Hermitian matrices. 
The basis element $g_j$ is Hermitian since
\begin{align}\label{eq:Hbasis}
	g_j^{\dagger} := \overline{g}_j^{\top} &= (-i)^{\lambda_j}~ Z(\bb_j)^{\top}~X(\ba_j)^{\top}
	= (-i)^{\lambda_j}~ Z(\bb_j)~ X(\ba_j) \notag\\
	& = i^{\lambda_j}~ (-1)^{\ba_j \cdot \bb_j}~ (-1)^{\ba_j \cdot \bb_j}~ X(\ba_j) ~Z(\bb_j) = i ^{\lambda_j} ~ X(\ba_j) ~Z(\bb_j).  
\end{align}

\begin{theorem}\label{thm:symplectic}
	Let $C$ be an $n-k$-dimensional self-orthogonal\index{symplectic!self-orthogonal} subspace of $\mathbb{F}_2^{2n}$ under the symplectic inner product\index{symplectic!inner product|seealso{inner product, symplectic}}\index{inner product!symplectic}. Let 
	$d:= w_Q\left(C^{\perp_s} \setminus C\right) = \min \{w_Q(\bv) : \bv \in C^{\perp_s} \setminus C \}$. Then there is an $\dsb{n,k,d}$-qubit\index{qubit} \emph{stabilizer} code\index{quantum code!stabilizer} $Q$.
\end{theorem}
\begin{proof}
	We lift $C:= \overline{G} \in \overline{\EE_n}$ to an abelian subgroup $G$ of $\EE_n$, with $G \cong \mathbb{F}_2^{n-k}$. Then $\mathbb{C}^{2^n} = \bigoplus_{\chi \in \widehat{G}} Q(\chi)$, where $Q(\chi)= L_{\chi} \mathbb{C}^{2^n}$, is the subspace 
	\[
	\{ \ket{\bv} \in \mathbb{C}^{2^n} : g \ket{\bv} = \chi(g) \, \ket{\bv} \mbox{ for all } g \in G\}.
	\]
	Showing that each $Q(\chi)$ is an $\dsb{n,k,d}$-qubit\index{qubit} code means proving 
	\[
	\dim_{\mathbb{C}} Q(\chi) = 2^{k} \mbox{ and } 
	d(Q(\chi)) \geq w_{Q} \left(C^{\perp_s} \setminus C \right).
	\]
	
	Consider the action of $\EE_n$ on $\{Q(\chi) :\chi \in \widehat{G}\}$. For any $\ket{\bv} \in Q(\chi)$ and $g \in G$, we have 
	$g \ket{\bv} = \chi(g) \, \ket{\bv}$. Thus, for any $\Ee \in \EE_n$ and any $g \in G$, 
	\[
	g (\Ee \ket{\bv}) =(-1)^{\inner{\overline{g},\overline{\Ee}}_s} \, 
	\Ee (g \ket{\bv}) =(-1)^{\inner{\overline{g},\overline{\Ee}}_s} \, \chi(g) \, \Ee \ket{\bv}. 
	\]
	Since $\chi_{\overline{\Ee}} : G \mapsto \{\pm 1\}$ and 
	$\chi_{\overline{\Ee}} (g) = (-1)^{\inner{\overline{g},\overline{\Ee}}_s}$ is a character\index{character} of $G$, we have
	\[
	g (\Ee \ket{\bv}) = \chi_{\overline{\Ee}}(g) \, \chi(g) \, \Ee \ket{\bv} = \chi'(g) \, \Ee \ket{\bv}\mbox{, for all }g \in G.
	\]
	This implies $\Ee \ket{\bv} \in Q(\chi')$ and $\Ee: Q(\chi) \mapsto Q(\chi')$, where $\chi' := \chi_{\overline{\Ee}} \, \chi$. Since $\EE_n$ is a group\index{group}, $\Ee$ is a bijection, making $\dim_{\mathbb{C}} Q(\chi) = \dim_{\mathbb{C}} Q(\chi')$. As $\Ee$ runs through $\EE_n$, $\chi_{\overline{\Ee}}$ takes all characters\index{character} of $G$, ensuring that $\dim_{\mathbb{C}} Q(\chi)$ is the same for all $\chi \in \widehat{G}$. Thus, $\dim_{\mathbb{C}} Q(\chi) =2^{n-(n-k)}=2^k$ for any $\chi \in \widehat{G}$.
	
	We now show that, if $\Ee \in \EE_{d-1}$ and $\ket{\bv_1}, \ket{\bv_2} \in Q(\chi)$ with $\braket{\bv_1}{\bv_2}=0$, then $\bra{\bv_1} \Ee_1 \, \Ee_2 \ket{\bv_2}=0$, where $\Ee:=\Ee_1 \, \Ee_2$. If $\overline{\Ee} \in \overline{G} = C$, then $\bra{\bv_1} \Ee \ket{\bv_2} = \chi(\overline{\Ee}) (\braket{\bv_1}{\bv_2})=0$. Otherwise, $\overline{\Ee} \notin C$. From $w_{Q}(\overline{\Ee})=w_{Q}(\Ee) \leq d-1$ and the assumption $\left(C^{\perp_s} \setminus C\right) \cap \overline{\EE}_{d-1} = \emptyset$, we know $\overline{\Ee} \notin C^{\perp_s}$. Hence, there exists $\overline{\Ee}' \in \overline{G}$ such that $\Ee \, \Ee' = - \Ee' \, \Ee$. Then, for $\ket{\bv_2} \in Q(\chi)$, we have 
	\[
	\Ee' \, \Ee \ket{\bv_2} = - \Ee \, \Ee' \ket{\bv_2} = - \chi(\Ee') \Ee \ket{\bv_2} \mbox{, with } 
	- \chi(\Ee') \neq \chi(\Ee').
	\]
	Therefore, $\Ee \ket{\bv_2} \in Q(\chi')$, with $\chi' \neq \chi$. Since $\ket{\bv_1} \in Q(\chi)$ and $Q(\chi)$ is orthogonal to $Q(\chi')$, we confirm $\bra{\bv_1} \Ee \ket{\bv_2}=0$.
\end{proof}

\begin{example}\label{ex:512}
	We exhibit a $\dsb{5,1,3}$-qubit\index{qubit} stabilizer code\index{quantum code!stabilizer} $Q$. Consider a subspace $C \subset \mathbb{F}_2^{10}$ with generator matrix 
	\[
	\begin{bmatrix}
		1 & 1 & 0 & 0 & 0 & 0 & 0 & 1 & 0 & 1 \\
		0 & 1 & 1 & 0 & 0 & 1 & 0 & 0 & 1 & 0 \\
		0 & 0 & 1 & 1 & 0 & 0 & 1 & 0 & 0 & 1 \\
		0 & 0 & 0 & 1 & 1 & 1 & 0 & 1 & 0 & 0   
	\end{bmatrix}=
	\begin{bmatrix}
		\bv_1\\
		\bv_2\\
		\bv_3\\
		\bv_4 
	\end{bmatrix}.
	\]
	One reads $\bv_1=(\ba|\bb)$ as having $\ba=(1,1,0,0,0)$ and $\bb=(0,0,1,0,1)$. The code $C$ is symplectic self-orthogonal\index{symplectic!self-orthogonal}, with $\dim_{\mathbb{F}_2} C=4 $ and $\dim_{\mathbb{F}_2} C^{\perp} =6$, \ie, the codimension is $2$. To extend the basis for $C$ to a basis for $C^{\perp_s}$ we use $(0,0,0,0,0,\,1,1,1,1,1)$ and $(1,1,1,1,1,\,0,0,0,0,0)$. Since $w_{Q}(C)=4$ and $w_{Q}(C^{\perp_s})=3$, one obtains $w_Q(C^{\perp_s} \setminus C)=3$. We can write the $\dsb{5,1,3}$-code $Q= Q(\chi_0)$ explicitly by using $G=\left\langle g_1,g_2,g_3,g_4\right\rangle$, with $g_1 =\sigma_x \otimes \sigma_x \otimes \sigma_z \otimes I_2 \otimes \sigma_z$, 
	$g_2 =\sigma_z \otimes \sigma_x \otimes \sigma_x \otimes \sigma_z \otimes I_2$, 
	$g_3 =I_2 \otimes \sigma_z \otimes \sigma_x \otimes \sigma_x \otimes \sigma_z$, and 
	$g_4 =\sigma_z \otimes  I_2 \otimes \sigma_z \otimes \sigma_x \otimes \sigma_x$.
	
	Since $k=1$ and $\dim_{\mathbb{C}} Q = 2^k=2$, two independent vectors in $\mathbb{C}^{32}$ form a basis of 
	$Q = \{ \ket{\bv} \in \mathbb{C}^{32} : g \ket{\bv} = \chi_0 (g) \ket{\bv} 
	\, \forall g \in G \}$. $Q$ consists of vectors which are fixed by all $g \in G$. After some computation, we conclude that $Q$ can be generated by $\ket{\bv_0}=\sum_{g \in G} g \ket{00000}$ and $\ket{\bv_1}=\sum_{g \in G} g \ket{11111}$. \quad \qedsymbol
\end{example}

With minor modifications, the qubit stabilizer formalism\index{quantum stabilizer formalism} extends to the general qudit case. A complete treatment is available in~\cite{Ketkar2006}. We outline the main steps here. An \textbf{$n$-qudit system}\index{qudit} is a nonzero element in $\left(\mathbb{C}^q\right)^{\otimes n} \cong \mathbb{C}^{q^n}$. Let $\ba=(a_1,\ldots,a_{n}) \in \mathbb{F}_q^n$. The standard $\mathbb{C}$-basis is 
\[
\{\ket{a_1 a_2 \ldots a_{n}} :=
\ket{a_1} \otimes \ket{a_2} \otimes \ldots \otimes \ket{a_{n}}: \ba \in \mathbb{F}_q^n\} 
\]
and an arbitrary vector in $\mathbb{C}^{q^n}$ is written $\ket{\psi} = \sum_{\ba \in \mathbb{F}_q^n} c_{\ba} \ket{\ba}$, with $c_{\ba} \in \mathbb{C}$ and $q^{-n} \sum_{\ba \in \mathbb{F}_q^n} \norm{c_{\ba}}^2=1$.

Let $\ba=(a_1,\ldots,a_n), \bb=(b_1,\ldots,b_n) \in \mathbb{F}_q^n$, and $\omega := e^{\frac{2 \pi i}{p}}$, where $q=p^m$. The error operators\index{quantum error operators|seealso {Pauli matrices}} form $\EE_n:= \{ \omega^{\beta}  X(\ba) Z(\bb) : \ba, \bb \in \mathbb{F}_q^n, \, \beta \in \mathbb{F}_p\}$ of cardinality $p q^{2n}$. The respective actions of $X(\ba)$ and $Z(\bb)$ on $\ket{\bv} \in \mathbb{C}^{q^n}$ are 
$X(\ba) \ket{\bv} = \ket{\ba + \bv}$ and 
$Z(\bb) \ket{\bv} = (\omega)^{\Tr (\bb \cdot \bv)} \ket{\bv}$. Hence, for 
$\Ee := \omega^{\beta} X(\ba) Z(\bb)$ and $\Ee' := \omega^{\beta'} X(\ba') Z(\bb')$ in $\EE_n$, one gets $\Ee \, \Ee' = \omega^{\Tr (\bb \cdot \ba' - \bb' \cdot \ba)} \, \Ee' \, \Ee$. The \textbf{symplectic weight}\index{symplectic!weight} of $(\ba|\bb)$ is the quantum weight\index{quantum weight} of $\Ee$. 

The \textbf{(trace) symplectic inner product}\index{symplectic!inner product|seealso{inner product, symplectic}}\index{inner product!symplectic} of $(\ba|\bb)$ and $(\ba'|\bb')$ in $\mathbb{F}_q^{2n}$ is
\begin{equation}\label{eq:symplectic}
	\inner{(\ba|\bb), (\ba'|\bb')}_s = \Tr(\bb \cdot \ba' - \bb' \cdot \ba). 
\end{equation}
The \textbf{symplectic dual} of $C \subseteq \mathbb{F}_q^{2n}$ is 
$C^{\perp_s} = \{ \bu \in \mathbb{F}_q^{2n} : \inner{\bu,\bc}_s=0 \, \forall \, \bc \in C \}$. As in the qubit\index{qubit} case, in the general qudit\index{qudit} setup, a subgroup $G$ of $\EE_n$ is abelian if and only if $\overline{G}$ is a symplectic self-orthogonal\index{symplectic!self-orthogonal} subspace of $\overline{\EE_n} \cong \mathbb{F}_q^{2n}$. The analogue of Theorem~\ref{thm:symplectic} follows.

\begin{theorem}\label{thm:qudit_sympl}
	Let $C$ be an $n-k$-dimensional self-orthogonal subspace of $\mathbb{F}_q^{2n}$ under the (trace) symplectic inner product\index{symplectic!inner product|seealso{inner product, symplectic}}\index{inner product!symplectic}. Let $d:= w_Q\left(C^{\perp_s} \setminus C\right) = \min \{w_Q(\bv) : \bv \in C^{\perp_s} \setminus C \}$. Then there is an $\dsb{n,k,d}$-qudit\index{qudit} \emph{stabilizer}\index{quantum code!stabilizer} code $Q$.
\end{theorem}

\begin{VF}
	``With group\index{group} and eigenstate\index{eigenstate}, we've learned to fix\\
	Your quantum errors with our quantum tricks.''
	
	\VA{Daniel Gottesman}{in {\it Quantum Error Correction Sonnet}}
\end{VF}

\section{Constructions via Classical Codes}

Any stabilizer code\index{quantum code!stabilizer} $Q$ is fully characterized by its stabilizer group\index{group!stabilizer}, that specifies the set of errors that $Q$ can correct. Any linear combination of the operators in the error set is correctable, allowing $Q$ to correct a continuous set of operators. For this reason, the best-known qubit\index{qubit} codes in the online table~\cite{Grassl} maintained by M. Grassl are given in terms of their stabilizer generators. The stabilizer approach has massive advantages over other frameworks, some of which will be mentioned below. It describes a large set of QECs, complete with their encoding and decoding mechanism, in a very compact form. 

A valid codeword of $Q$ is a $+ 1$ \textbf{eigenvector}\index{eigenvector} of \emph{all} the stabilizer generators. An error $\Ee$, expressed as a tensor product\index{tensor product} of Pauli operators, anticommutes with some of the stabilizer generators and commutes with others. It sends a codeword to an \textbf{eigenstate}\index{eigenstate} of the stabilizer generators. The \textbf{eigenvalue} remains $+1$ for all operators that commute with $\Ee$ but becomes $-1$ for those generators that anticommute with $\Ee$. From the resulting error syndrome, one knows which Pauli operators acts on which qubits. Applying the respective Pauli operators on the corresponding locations corrects the error. Suppose that the location of error is known, but the type is not, then this is a \textbf{quantum erasure}\index{quantum erasure}. By the Knill-Laflamme condition\index{Knill-Laflamme condition}, correcting $\ell$ general errors means correcting $2 \ell$ erasures. 

The encoding and syndrome reading circuits can be written using only three quantum gates, namely the \textbf{Hadamard gate}\index{quantum gate!Hadamard}, the \textbf{phase $S$ gate}\index{quantum gate!$S$}, and the \textbf{CNOT gate}\index{quantum gate!CNOT}, whose respective matrices are
\[
{\rm H}=\frac{1}{\sqrt{2}} 
\begin{bmatrix}
	1&1\\1&-1
\end{bmatrix}, \quad 
{\rm S} = \begin{bmatrix}
	1 & 0 \\ 0 & i
\end{bmatrix} , \quad
{\rm CNOT} = \begin{bmatrix}
	1 & 0 & 0 & 0 \\ 0 & 1 & 0 & 0\\ 0&0&0&1 \\ 0&0&1&0
\end{bmatrix}.
\]
A treatment on the circuit implementations is available, \eg, in~\cite{Grassl2011}.

We now look into suitable classical codes that fully describe the set of correctable errors. All constructions are applications of the stabilizer formalism\index{quantum stabilizer formalism}. Since all of the inner products used are nondegenerate, \ie, $(C^{\perp})^{\perp} = C$, one can interchange self-orthogonality and dual-containment, provided that the derived parameters are adjusted accordingly.   

First, we consider a generic construction of $q$-ary quantum codes\index{quantum code} via additive (\ie, $\mathbb{F}_q$-linear) codes over $\mathbb{F}_{q^2}$. Let $\{1,\gamma\}$ be a  basis of $\mathbb{F}_{q^2}$ over $\mathbb{F}_q$. 
The map $\Phi: \overline{\EE_n} \cong \mathbb{F}_q^{2n} \mapsto \mathbb{F}_{q^2}^n$ that sends 
$\overline{\Ee}:= \bv =(\ba | \bb ) = (a_1, \ldots, a_n | b_1,\ldots, b_n)$ to 
$(a_1 + \gamma b_1, \ldots, a_n + \gamma b_n)$ is an isomorphism of $\mathbb{F}_q$-vector spaces. It is also an isometry, since $\wQ(\Ee)=\wH(\Phi((\ba|\bb )))$. For any $\bu:=(u_1,\ldots,u_n)$ and 
$\bv:=(v_1, \ldots,v_n) \in \mathbb{F}_{q^2}^n$, we use $\bu^q$ to denote $(u_1^q,\ldots,u_n^q)$ and define the \textbf{trace alternating inner product}\index{inner product!trace alternating} of $\bu$ and $\bv$ as
$\displaystyle{
	\inner{\bu,\bv}_{\rm alt} := \Tr_{\mathbb{F}_{q^2}/\mathbb{F}_q} \left(\frac{\bu \cdot \bv^q - \bu^q \cdot \bv}{\gamma - \gamma^q }\right)}$. When $q=2$, $\inner{\bu,\bv}_{\rm alt}$ coincides with the \textbf{trace Hermitian inner product}\index{inner product!trace Hermitian} $\inner{\bu,\bv}_{\Tr_{\HH}} := \sum_{j=1}^n \left(u_j v_j^2 + u_j^2 v_j\right)$, since $\gamma - \gamma^2 = 1$. For any $(\ba,\bb)$ and $(\ba',\bb')$ in $\overline{\EE_n}$, we immediately verify that 
$\inner{(\ba,\bb), (\ba',\bb')}_s = \inner{\Phi((\ba,\bb)),\Phi((\ba',\bb')}_{\rm alt}$. Hence, a linear code $C \subseteq \mathbb{F}_q^{2n}$ is symplectic self-orthogonal\index{symplectic!self-orthogonal} if and only if the additive code $\Phi(C)$ is trace alternating self-orthogonal. The \textbf{Hermitian inner product}\index{inner product!Hermitian} of any $\bu,\bv \in \mathbb{F}_{q^2}^n$ is $\inner{\bu,\bv}_{\HH} := \sum_{j=1}^n u_j v_j^q$. If $\Phi(C)$ is $\mathbb{F}_{q^2}$-linear, instead of being strictly additive, then $\Phi(C) \subseteq (\Phi(C))^{\perp_{\rm alt}}$ if and only if $\Phi(C) \subseteq (\Phi(C))^{\perp_{\HH}}$. Thus, Theorem~\ref{thm:qudit_sympl} has the following equivalent statement.
\begin{theorem}\label{thm:additive}
	Let $\mathcal{C} \subseteq \mathbb{F}_{q^2}^n$ be an $\mathbb{F}_q$-additive code such that $\mathcal{C} \subseteq \mathcal{C}^{\perp_{\rm alt}}$, with $\abs{\mathcal{C}}=q^{n-k}$. Then there exists an $\dsb{n,k,d}_q$ quantum code\index{quantum code} $Q$ with 
	\[
	d(Q) = \wH(\mathcal{C}^{\perp_{\rm alt}} \setminus \mathcal{C}) = 
	\min \{\wH(\bv) : \bv \in \mathcal{C}^{\perp_{\rm alt}}  \setminus \mathcal{C} \}.
	\]
	If $\mathcal{C}$ is $\mathbb{F}_{q^2}$-linear, we can conveniently replace the trace alternating inner product\index{inner product!trace alternating} by the Hermitian inner product\index{inner product!Hermitian}, which is easier to compute.
\end{theorem}

If $C$ is $\mathbb{F}_4$-additive and is even, \ie, $\wH(\bc)$ is even for all $\bc \in C$, then $C$ is trace Hermitian self-orthogonal. If $C$ is trace Hermitian self-orthogonal and $C$ is $\mathbb{F}_4$-linear, then $C$ is an even code. 

The quantum codes\index{quantum code} in Theorem~\ref{thm:additive} are modeled after classical codes with an additive structure, but the error operators\index{quantum error operators|seealso {Pauli matrices}} are in fact multiplicative. An error $\Ee$ may have the same effect as $\Ee {\rm S}$ where ${\rm S} \neq I$ is an element of the stabilizer group\index{group!stabilizer}. A QEC is degenerate or \textbf{impure}\index{quantum code!impure} if the set of correctable errors contains degenerate errors. Studies on impure codes has been rather scarce. The existence of two inequivalent $\dsb{6,1,3}$ impure qubit codes was shown in~\cite[Section IV]{Calderbank1998}. Remarkably, there is no $\dsb{6,1,3}$ pure\index{quantum code!pure} qubit\index{qubit} code. A systematic construction based on \textbf{duadic codes}\index{duadic code} and further discussion on the advantages of degenerate quantum codes\index{quantum code} are supplied in~\cite{Aly2006}.

A very popular construction is based on nested classical codes. We denote the Euclidean dual of $C$ by $C^{\perp_{\rm{E}}}$. 
\begin{theorem}[\textbf{Calderbank-Shor and Steane (CSS) Construction}]\index{quantum code!Calderbank-Shor-Steane (CSS)}\label{thm:CSS} 
	Let $C_{j}$ be an $[n,k_{j},d_{j}]_{q}$-code for $j\in \{1,2\}$ with $C_{1}^{\perp_{\Euc}} \subseteq C_{2}$. Then there is an 
	\[
	\dsb{n,k_{1}+k_{2}-n,\min \{\wH(C_{2} \setminus C_{1}^{\perp_{\Euc}}),\wH(C_{1} \setminus C_{2}^{\perp_{\Euc}}) \}}_{q}\mbox{-code }Q.
	\]
	The code is pure\index{quantum code!pure} whenever 
	$\min \{\wH(C_{2} \setminus C_{1}^{\perp_{\Euc}}),\wH(C_{1} 
	\setminus C_{2}^{\perp_{\Euc}}) \} = \min \{ d_1,d_2\}$.
\end{theorem}
\begin{proof}
	Let $G_j$ and $H_j$ be the generator and parity-check matrices of $C_j$. Consider the linear code $C \subseteq \mathbb{F}_q^{2n}$ with generator matrix 
	$\displaystyle{
		\begin{bmatrix}
			H_1 & \0 \\
			\0 & H_2
		\end{bmatrix}
	}$. Since $C_1^{\perp_{\Euc}} \subseteq C_2$, we have $H_1 H_2^{\top} = \0$. Similarly, from $C_2^{\perp_{\Euc}} \subseteq C_1$, we know $H_2 H_1^{\top} = \0$. Define $C^{\perp_s}$ to be the code with parity check and generator matrices, respectively,  
	$\displaystyle{
		\begin{bmatrix}
			H_2 & \0 \\
			\0 & H_1
	\end{bmatrix}}$ and 
	$\displaystyle{
		\begin{bmatrix}
			G_2 & \0 \\
			\0 & G_1
	\end{bmatrix}}$. We verify that $C \subseteq C^{\perp_s}$, with $\dim_{\mathbb{F}_q} C = 2n-(k_1 + k_2)$. By Theorem~\ref{thm:qudit_sympl}, $\dim_{\mathbb{C}} Q= n - \dim_{\mathbb{F}_q} C= k_1 + k_2 -n$. The distance computation is clear.
\end{proof}

A special case of the \textbf{CSS construction}\index{quantum code!Calderbank-Shor-Steane (CSS)} comes via a Euclidean dual-containing code $C^{\perp_{\rm{E}}} \subseteq C$. From such an $[n,k,d]_q$-code $C$, one obtains an 
$\dsb{n, 2k-n,\geq d}_q$-code $Q$. The next method allows for most qubit CSS codes to be enlarged while avoiding a significant drop in the distance. The choice of the extra vectors in the generator matrix of $C'$ is detailed in~\cite[Section III]{Steane1999}. 
\begin{theorem}[\textbf{Steane Enlargement of CSS Codes}]\index{quantum code!Steane enlargement of CSS}
	Let $C$ be an $[n,k,d]_2$-code that contains its Euclidean dual $C^{\perp_{\Euc}} \subseteq C$. Suppose that $C$ can be enlarged to $C'=[n,k' > k+1 ,d']_2$. Then there exists a pure\index{quantum code!pure} qubit code of parameters $\dsb{n,k+k'-n, \min\{d, \lceil{3d'/2 \rceil}\}}$.
\end{theorem}
A generalization to the qudit\index{qudit} case was subsequently given in~\cite{Ling2010}, where the distance is $\min\{d, \lceil{\frac{q+1}{q} \,d' \rceil}\}$. Comparing the minimum distances in the resulting codes, the enlargement offers a better chance of relative gain in the qubit\index{qubit} case as compared with the $q>2$ cases.

Lison{\v{e}}k and Singh, inspired by the classical \textbf{Construction X}, proposed a modification to qubit\index{qubit} stabilizer codes\index{quantum code!stabilizer} in~\cite{Lisonek2014}. The construction generalizes naturally to qudit\index{qudit} codes.

\begin{theorem}[\textbf{Quantum Construction X}]\index{quantum code!Construction X of}\label{thm:X}
	For an $[n,k]_{q^2}$-linear code $C$, let $e:=k - \dim(C \cap C^{\perp_{\HH}})$. Then there exists an $\dsb{n+e,n-2k+e,d}_q$-code $Q$, with $d:=d(Q) \geq \min \{d(C^{\perp_{\HH}}), d(C + C^{\perp_{\HH}})+1\}$, where 
	$C + C^{\perp_{\HH}}:=\{\bu + \bv : \bu \in C, \bv \in C^{\perp_{\HH}}\}$. 
\end{theorem}
The case $e=0$ is the usual stabilizer construction. To prevent a sharp drop in $d$, we want small $e$, \ie, large \textbf{Hermitian hull}\index{Hermitian hull} $C \cap C^{\perp_{\HH}}$. 

We shift our attention now to \textbf{propagation rules}\index{quantum propagation rules} and \textbf{bounds}\index{quantum bound}. Most of them are direct consequences of the propagation rules and bounds on the classical codes used as ingredients in the above constructions.

\begin{proposition}[\textbf{{see \cite[Theorem 6]{Calderbank1998} for the binary case}}] \label{propagation}
	From an $\dsb{n,k,d}_q$-code, the following codes can be derived: an $\dsb{n,k-1, \geq d}_q$-code by \textbf{subcode construction}\index{quantum!code!subcode construction of}, an $\dsb{n+1,k, \geq d}_q$-code by \textbf{lengthening}\index{quantum!code!lengthening of}, and an $\dsb{n-1,k, \geq d-1}_q$-code by \textbf{puncturing}\index{quantum!code!puncturing of}.
\end{proposition} 
The analogue of \textbf{shortening}\index{quantum!code!shortening of} is less straightforward. It requires the construction of an auxiliary code and, then, a check on whether this code has codewords of a given length. The details on how to shorten quantum codes\index{quantum code} are available in~\cite[Section 4]{GBR04}, building upon the initial idea of Rains in~\cite{Rains99}.

How can we measure the goodness of a QEC? For stabilizer codes\index{quantum code!stabilizer}, given their classical ingredients and constructions, there are numerous bounds\index{quantum bound}.

\begin{theorem}[\textbf{Quantum Hamming Bound\index{quantum bound!Hamming}, see {\cite{Calderbank1998}} for the binary case}] 
	Let $Q$ be a pure\index{quantum code!pure} $\dsb{n,k,d}_q$-code with $d \geq 2 \ell + 1 $ and $k>0$. Then
	\begin{equation}
		q^{n-k} \geq \sum_{j=0}^{\ell} (q^2-1)^j \binom{n}{j}.
	\end{equation}
	$Q$ is \textbf{perfect}\index{quantum code!perfect} if it meets the bound.
\end{theorem}
The proof comes from the observation that   
\[
q^n \geq \sum_{\overline{\Ee} \in \overline{\EE_n}(\ell)} 
\dim_{\mathbb{C}} (\overline{\Ee} \, Q)=
\dim_{\mathbb{C}} Q \cdot \abs{\overline{\EE_n}(\ell)}
=q^k \sum_{j=0}^{\ell} (q^2-1)^j \binom{n}{j}.
\]
The code in Example~\ref{ex:512} is perfect\index{quantum code!perfect}. It has $2^{n-k}=16 = \sum_{j=0}^{1} 3^j \binom{5}{j}=1+15$.

Here is another bound which had been established as a necessary condition for pure\index{quantum code!pure} stabilizer codes\index{quantum code!stabilizer}.
\begin{theorem}[\textbf{Quantum Gilbert-Varshamov Bound\index{quantum bound!Gilbert-Varshamov} {\cite{Feng2004}}}]\label{thm:GV}
	Let $n > k \geq 2$, $d \geq 2$, and $n \equiv k \pmod{2}$. A pure\index{quantum code!pure} $\dsb{n,k,d}_q$-code exists if 
	\[
	\frac{q^{n-k+2}-1}{q^2-1} > \sum_{j=1}^{d-1} (q^2-1)^{j-1} \binom{n}{j}.
	\]
\end{theorem}

An upper bound\index{quantum bound}, which is well-suited for computer search, is the \textbf{quantum Linear Programming (LP) bound}\index{quantum bound!Linear Programming (LP)}. In the qubit\index{qubit} case, the bound is explained in details in~\cite[Section VII]{Calderbank1998}. The same routine adjusts immediately to the the general qudit\index{qudit} case, as was shown in~\cite[Section VI]{Ketkar2006}. The main tools are the \textbf{MacWilliams identities}~\cite{Macwilliams1963}\index{weight distribution!MacWilliams identities}. These are linear relations between the \textbf{weight distribution} of a classical code and its dual. They hold for all of the inner products we are concerned with here and have been very useful in ruling out the existence of quantum codes of certain ranges of parameters. Rains supplied a nice proof for the next bound, which is a corollary to the quantum LP bound\index{quantum bound!Linear Programming (LP)}, using the \textbf{quantum weight enumerator}\index{quantum weight enumerator}\index{weight enumerator} in~\cite{Rains99}. A quantum code\index{quantum code} that reaches the equality in the bound is said to be \textbf{quantum MDS (QMDS)}\index{quantum code!Maximum distance separable (QMDS)}.

\begin{theorem}[\textbf{Quantum Singleton Bound}\index{quantum bound!Singleton}]\label{th:QSingleton} An $\dsb{n,k,d}_q$-code with $k > 0$ satisfies $k \leq n-2d+2$.
\end{theorem}

Nearly all known families of classical codes over finite fields, especially those with well-studied algebraic and combinatorial structures, have been used in each of the constructions above. A partial list, compiled in mid 2005 as~\cite[Table II]{Ketkar2006}, already showed a remarkable breadth. The large family of cyclic-type codes, whose corresponding structures in the rings of polynomials are \textbf{ideals}\index{ring!ideal}, has been a rich source of ingredients for QECs with excellent parameters. This includes the \textbf{BCH}\index{BCH code}, \textbf{cyclic}\index{cyclic code}, \textbf{constacyclic}\index{constacyclic code}, \textbf{quasi-cyclic}\index{quasi-cyclic code}, and \textbf{quasi-twisted} codes\index{quasi-twisted code}. In the family, the nestedness property, very useful in the CSS construction\index{quantum code!Calderbank-Shor-Steane (CSS)}, comes for free. A great deal is known about their dual codes under numerous applicable inner products. For small $q$, the structures allow for extensive computer algebra searchers for reasonable lengths, aided by their \textbf{minimum distance bounds}\index{code!minimum distance}\index{bound}. 

The most comprehensive record for best-known qubit\index{qubit} codes is Grassl's online table~\cite{Grassl}. Numerous entries have been certified optimal, while still more entries indicate gaps between the best-possible and the best-known. It is a two-fold challenge to contribute meaningfully to the table. First, for $n \leq 100$, many researchers have attempted exhaustive searches. Better codes are unlikely to be found without additional clever strategies. Second, for $n > 100$, computing the actual distance $d(Q)$ tends to be prohibitive. As the length and dimension grow, computing the minimum distances of the relevant classical codes to derive the quantum distance\index{quantum code!minimum distance of} is hard~\cite{Vardy1997}. Improvements remain possible, with targeted searches. Recent examples include the works of Galindo \etal on quasi-cyclic constructions of quantum codes\index{quantum code}~\cite{Galindo2018}, where Steane enlargement\index{quantum code!Steane enlargement of CSS} is deployed, the search reported in~\cite{Lisonek2014} on cyclic codes\index{cyclic code} over $\mathbb{F}_{4}$, where Construction X\index{quantum code!Construction X of} is used with $e \in \{1,2,3\}$, and similar random searches on quasi-cyclic codes\index{quasi-cyclic code} done in~\cite{Ezerman2019} for qubit\index{qubit} and qutrit\index{qutrit} codes.

Less attention has been given to record-holding qutrit codes, for which there is no publicly available database of comparative extense. A table listing numerous qutrit\index{qutrit} codes is kept by Y. Edel in~\cite{Edel} based on their explicit construction as \textbf{quantum twisted codes}\index{quantum code!twisted} in~\cite{Bierbrauer2000}. Better codes than many of those in the table have since been found. 

Attempts to derive new quantum codes\index{quantum code} by shortening good stabilizer codes\index{quantum code!stabilizer} motivate closer studies on the weight distribution of the classical auxiliary codes, in particular when the stabilizer codes\index{quantum code!stabilizer} are QMDS. Shortening is very effective in constructing qudit\index{qudit} MDS codes of lengths up to $q^2+2$ and minimum distances\index{quantum code!minimum distance of} up to $q+1$. 

There has been a large literature on QMDS\index{quantum code!Maximum distance separable (QMDS)}. All of the above constructions via classical codes as well as the propagation rules\index{quantum propagation rules} have been applied to families of classical MDS codes, particularly the \textbf{Generalized Reed-Solomon}\index{Generalized Reed-Solomon code} and the \textbf{constacyclic MDS codes}\index{constacyclic code}\index{maximum distance separable (MDS) code}. Since the dual of an MDS code is MDS, the dual distance is evident, leaving only the orthogonality property to investigate. While the theoretical advantages are clearly abundant, there are practical limitations. The length of such codes is bounded above by $q^2+2$, when $q$ is even, and by $q^2+1$, when $q$ is odd, assuming the \textbf{MDS conjecture}\index{conjecture!MDS}.

For qubit\index{qubit} codes, the only nontrivial QMDS\index{quantum code!Maximum distance separable (QMDS)} are those with parameters $\dsb{5,1,3}$, $\dsb{6,0,4}$, and $\dsb{2m,2m-2,2}$. As $q$ grows larger, the practical value of QMDS codes quickly diminishes, since controlling qudit\index{qudit} systems with $q>3$ is currently prohibitive. A list for $q$-ary QMDS\index{quantum code!Maximum distance separable (QMDS)} codes, with $2 \leq q \leq 17$, is available in~\cite{Grassl2015}. Another list that covers families of QMDS codes and their references can be found in~\cite[Table V]{Chen2015}. More works in QMDS\index{quantum code!Maximum distance separable (QMDS)} continue to appear, with detailed analysis on the self-orthogonality conditions supplied from number theoretic and combinatorial tools.

Taking \textbf{algebraic geometry (AG) codes}\index{Algebraic Geometry (AG) code} as the classical ingredients is another route. A wealth of favourable results had already been available prior to the emergence of QECs. Curves with many \textbf{rational points} often lead, via the \textbf{Goppa construction}\index{Algebraic Geometry (AG) code!Goppa construction of}, to codes with good parameters. Their duals are well-understood, via the \textbf{residue formula}. Their designed distances can be computed from the \textbf{Riemann-Roch theorem}\index{Riemann-Roch theorem}. Chen \etal showed how to combine Steane enlargement\index{quantum code!Steane enlargement of CSS} and concatenated AG codes\index{Algebraic Geometry (AG) code} to derive excellent qubit codes in~\cite{Chen2005}. A quantum asymptotic bound\index{quantum bound} was established in~\cite{Feng2006}. Construction of QECs from AG codes\index{Algebraic Geometry (AG) code} was initially a very active line of inquiry. It had somewhat lessened in the last decade, mostly due to lack of practical values to add to the quest as $q$ grows. 

Using \textbf{codes over rings}\index{linear code over a ring} to construct QECs have also been tried. This route, however, does not usually lead to parameter improvements over QECs constructed from codes over fields. The absence of a direct connection from codes over rings to QECs necessitates the use of the \textbf{Gray mapping}\index{Gray map}, which often causes a significant drop in the minimum distance\index{quantum code!minimum distance of}.

\section{Going Asymmetric}
So far we have been working on the assumption that the bit-flips\index{quantum error!bit-flip} and the phase-flips\index{quantum error!phase-flip} are equiprobable. Quantum systems, however, have noise properties that are dynamic and asymmetric. The fault-tolerant threshold is improved when asymmetry is considered~\cite{Aliferis2008}.  It was Steane who first hinted at the idea of adjusting error-correction to the particular characteristics of the 	channel in~\cite{Steane1996}. Designing error control methods to suit the noise profile, which can be hardware-specific, is crucial. The study of \textbf{asymmetric quantum codes (AQCs)}\index{quantum code!asymmetric (AQCs)} gained traction when the ratios of how often $\sigma_z$ occurs over the occurrence of $\sigma_x$ were discussed in~\cite{Ioffe2007}, with follow-up constructions offered soon after in~\cite{Sarvepalli2009}. Wang \etal~established a mathematical model of AQCs in the general qudit\index{qudit} system in~\cite{Wang2010}. 

As in the symmetric case, $\EE_n:= \{ \omega^{\beta}  X(\ba) Z(\bb) : \ba, \bb \in \mathbb{F}_q^n, \, \beta \in \mathbb{F}_p\}$. An error $\Ee := \omega^{\beta} X(\ba) Z(\bb) \in \EE_n$ has ${\rm wt_{X}}(\Ee):=\wH(\ba)$ and ${\rm wt_{Z}}(\Ee):=\wH(\bb)$.

\begin{definition}[\textbf{Asymmetric Quantum Codes}]\label{def:AQC} 
	Let $d_{x}$ and $d_{z}$ be positive integers. A qudit\index{qudit} code $Q$ with dimension $K \geq q$ is called an asymmetric quantum code (AQC)\index{quantum code!asymmetric (AQCs)} with parameters $((n,K,d_{z},d_{x}))_{q}$ or $\dsb{n,k,d_{z},d_{x}}_{q}$, where $k=\log_{q} K$, if $Q$ detects $d_{x}-1$ qudits of $X$-errors and, at the same time, $d_{z}-1$ qudits of $Z$-errors. The code $Q$ is pure\index{quantum code!pure} if $\ket{\varphi}$ and $\Ee \ket{\psi}$ are orthogonal for any $\ket{\varphi}, \ket{\psi} \in Q$ and any $\Ee \in \EE_{n}$ such that $1 \leq {\rm wt_{X}}(\Ee)\leq d_{x}-1$ or 
	$1 \leq {\rm wt_{Z}}(\Ee)\leq d_{z}-1$. Any code $Q$ with $K=1$ is assumed to be pure\index{quantum code!pure}. 
	
	An $\dsb{n,k,d,d}_{q}$-AQC is symmetric, with parameters $\dsb{n,k,d}_{q}$, but the converse is not true since, for $\Ee \in \EE_{n}$ with ${\rm wt_{X}}(\Ee)\leq d-1$ and ${\rm wt_{Z}}(\Ee)\leq d-1$, 
	$\wQ(\Ee)$ may be bigger than $d-1$.
\end{definition}

To date, most known families of AQCs come from the  asymmetric version of the CSS construction\index{quantum code!Calderbank-Shor-Steane (CSS)} and its generalization in~\cite{Ezerman2013}. 
\begin{theorem}[\textbf{CSS-like Constructions for AQCs}]\label{thm:CSSAQC} Let $C_{j}$ be an $[n,k_{j},d_{j}]_{q}$-code for $j \in \{1,2\}$. Let $C_{j}^{\perp_{*}}$ be the dual of $C_j$ under one of the Euclidean\index{inner product!Euclidean}, the trace Euclidean\index{inner product!trace Euclidean}, the Hermitian\index{inner product!Hermitian}, and the trace Hermitian\index{inner product!trace Hermitian} inner products. Let $C_{1}^{\perp_{*}} \subseteq C_{2}$, with 
	\begin{align*}
		d_{z} &:=\max \{ \wH(C_{2} \setminus C_{1}^{\perp_{*}}), \wH (C_{1} 
		\setminus C_{2}^{\perp_{*}}) \} \mbox{ and } \\
		d_{x} &:=\min \{ \wH(C_{2} \setminus C_{1}^{\perp_{*}}),
		\wH(C_{1} \setminus C_{2}^{\perp_{*}}) \}.
	\end{align*}
	Then there exists an $\dsb{n,k_{1}+k_{2}-n,d_{z},d_{x}}_{q}$-code $Q$, which is pure\index{quantum code!pure} whenever $\{ d_{z},d_{x} \} = \{ d_{1},d_{2} \}$. If we have $C \subseteq C^{\perp_{*}}$ where $C$ is an $[n,k,d]_{q}$-code, then $Q$ is an $\dsb{n,n-2k,d',d'}_{q}$-code, where $d' = \wH(C^{\perp_{*}} \setminus C)$. The code $Q$ is pure whenever $d'=d^{\perp_{*}}:=d(C^{\perp_{*}})$.
\end{theorem}

The propagation rules\index{quantum propagation rules} and bounds\index{quantum bound} for AQCs follow from the relevant rules and bounds on the nested codes and their respective duals. Details on how to derive new AQCs from already known ones were discussed by La Guardia in~\cite{Guardia2013}. The \textbf{asymmetric version of the quantum Singleton bound}\index{quantum bound!Singleton} reads $k \leq n-(d_x+d_x)+2$. To benchmark codes of large lengths, one can use the the \textbf{asymmetric versions of the quantum Gilbert-Varshamov bound}\index{quantum bound!Gilbert-Varshamov} established by Matsumoto in~\cite{Matsumoto2017}.   

Many best-performing AQCs with small $d_x$ and moderate $n$ and $k$ were derived in~\cite{Ezerman2013}. The optimal ones reach the upper bounds certified by an improved LP bound\index{quantum bound!Linear Programming (LP)}, called the \textbf{triangle bound}\index{quantum bound!triangle} in~\cite[Section V]{Ezerman2013}. Recent results, covering also AQCs of large lengths, came from an interesting family of nested codes defined from multivariate polynomials and Cartesian product point sets due to Gallindo \etal in~\cite{GGHR2018}. 

Most known \textbf{asymmetric quantum MDS (AQMDS) codes}\index{quantum code!asymmetric MDS (AQMDS)} are pure CSS\index{quantum code!Calderbank-Shor-Steane (CSS)}. Assuming the validity of the MDS conjecture, all possible parameters that pure\index{quantum code!pure} CSS AQMDS codes can have were established in~\cite{EJKL13}.
\begin{theorem}\label{thm:AQMDS}
	Assuming the MDS conjecture\index{conjecture!MDS}, there is a pure CSS\index{quantum code!Calderbank-Shor-Steane (CSS)} AQMDS code with parameters $[[n,j,d_z,d_x]]_q$, where $\{d_z,d_x \}=\{n-k-j+1,k+1\}$ if and only if one of the followings holds:
	\begin{enumerate}[noitemsep]
		\item $q$ is arbitrary, $n\ge 2$, $k\in\{1,n-1\}$, and $j\in\{0,n-k\}$.
		\item $q=2$, $n$ is even, $k=1$, and $j=n-2$.
		\item $q\ge 3$, $n\ge 2$, $k=1$, and $j=n-2$.
		\item $q\ge 3$, $2\le n\le q$, $k \le n-1$, and $0\le j\le n-k$.
		\item $q\ge 3$, $n=q+1$,  $k \le n-1$, and $j\in\{0,2,\ldots,n-k\}$.
		\item $q=2^m$, $n=q+1$, $j=1$, and $k\in \{2,2^{m}-2 \}$.		
		\item $q=2^m$ where $m\ge 2$, $n=q+2$,  
		$$ \left\{
		\begin{array}{l} k=1,\mbox{ and }j\in\{2,2^m-2\},\\
			k=3,\mbox{ and }j\in\{0,2^m-4,2^m-1\},\mbox{ or,}\\
			k=2^m-1,\mbox{ and }j\in\{0,3\}.\end{array} \right. $$
	\end{enumerate}
\end{theorem}

Going forward, three general challenges can be identified. First, find better AQCs, particularly in qubit\index{qubit} systems, than the currently best-known. More tools to determine or to lower bound $d_z$ and $d_x$ remain to be explored if we are to improve on the parameters. Second, construct codes with very high $d_z$ to $d_x$ ratio, since experimental results suggest that this is typical in qubit channels. Third, find conditions on the nested classical codes that yield impure\index{quantum code!impure} codes. 

\section{Other Approaches and a Conclusion}
We briefly mention other approaches to quantum error control before concluding. 

Successful small-scale hardware implementations often rely on \textbf{topological codes}\index{quantum code!topological}, first put forward by Kitaev~\cite{Kitaev2003}. This family of codes includes \textbf{surface codes}\index{quantum code!surface}~\cite{Bravyi2014} and \textbf{color codes}\index{quantum code!color}~\cite{Bombin2015}. Topological codes\index{quantum code!topological} encode information in, mostly $2$-dimensional, lattices. They are CSS\index{quantum code!Calderbank-Shor-Steane (CSS)} codes with a clever design. The lattice, on which the stabilizer generators act locally, has a bounded weight. The extra restrictions make the error syndrome easier to infer.

Instead of block quantum codes\index{quantum code}, studies have been done on \textbf{convolutional qubit codes}\index{qubit}\index{quantum code!convolutional}, see, \eg, \cite{Forney2007} and subsequent works that cited it. The logical qubits are encoded and transmitted as soon as they arrive in a steady stream. The rate $k$ over $n$ is fixed, but the length is not. This type of codes, like their classical counterparts, may be useful in \textbf{quantum communication}\index{quantum communication}.

An approach, that does not require self-orthogonality, constructs \textbf{entanglement-assisted quantum codes (EA-QECs)}\index{quantum code!entanglement-assisted (EA-QEC)}~\cite{Brun2006}. The price to pay is the need for a number of maximally entangled states, called \textbf{ebits} for \textbf{entangled qubits}, to increase either the rate or the ability to handle errors. Producing and maintaining ebits, however, tend to be costly, which offset their efficacy. Pairs of classical codes, whose intersections have some prescribed dimensions, were shown to result in EA-QECs in~\cite[Section 4]{Guenda2019}. A formula on the optimal number of ebits that an EA-QEC\index{quantum code!entanglement-assisted (EA-QEC)} requires is given in~\cite{Wilde2008}.
\begin{theorem}
	Given a linear $[n,k,d]_{q^2}$-code $C$ with parity check matrix $H$, the code $C^{\perp_{\rm H}}$ stabilizes an EA-QEC with parameters $\dsb{n,2k-n+c,d;c}_q$, where $c:=\rank(HH^{\dagger})$ is the number of ebits required.
\end{theorem}

A larger class of QECs that includes all stabilizer codes\index{quantum code!stabilizer} is the \textbf{codeword stabilized (CWS) codes}\index{quantum code!codeword stabilized (CWS)}. The framework was proposed by Cross \etal in~\cite{Cross2009} to unify stabilizer (additive) codes and known examples of good nonadditive codes. General constructions for large CWS\index{quantum code!codeword stabilized (CWS)} codes are yet to be devised. Also currently unavailable are efficient encoding and decoding algorithms.

The bridge between classical coding theory and quantum error control was firmly put in place via the stabilizer formalism\index{quantum stabilizer formalism}. Various generalizations and modifications have been studied since, benefitting both the classical and quantum sides of the error-control theory. Well-researched tools and the wealth of results in classical coding theory translate almost effortlessly to the design of good quantum codes\index{quantum code}, moving far beyond what is currently practical to implement in actual quantum devices. Research problems triggered by error-control issues in the quantum setup revive and expand studies on specific aspects of classical codes, which were previously overlooked or deemed not so interesting. This fruitful cross-pollination between the classical and the quantum, in terms of error control, is set to continue.

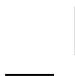
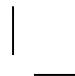
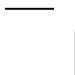
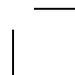
\printindex

\begin{thebibliography}{10}
	
	\bibitem{Aharonov2008}
	D.~Aharonov and M.~Ben-Or.
	\newblock Fault-tolerant quantum computation with constant error rate.
	\newblock {\em {SIAM} J. Comput.}, 38(4):1207--1282, Jan 2008.
	
	\bibitem{Aliferis2008}
	P.~Aliferis and J.~Preskill.
	\newblock Fault-tolerant quantum computation against biased noise.
	\newblock {\em Phys. Rev. A}, 78(5), Nov 2008.
	
	\bibitem{Aly2006}
	S.~A. Aly, A.~Klappenecker, and P.~K. Sarvepalli.
	\newblock Remarkable degenerate quantum stabilizer codes derived from duadic
	codes.
	\newblock In {\em Proc. Int. Symp. Inf. Theory (ISIT)}. {IEEE}, Jul 2006.
	
	\bibitem{Bierbrauer2000}
	J.~Bierbrauer and Y.~Edel.
	\newblock Quantum twisted codes.
	\newblock {\em J. Comb. Des.}, 8(3):174--188, 2000.
	
	\bibitem{Bombin2015}
	H.~Bomb{\'{\i}}n.
	\newblock Gauge color codes: optimal transversal gates and gauge fixing in
	topological stabilizer codes.
	\newblock {\em New J. Phys.}, 17(8):083002, Aug 2015.
	
	\bibitem{Bravyi2014}
	S.~Bravyi, M.~Suchara, and A.~Vargo.
	\newblock Efficient algorithms for maximum likelihood decoding in the surface
	code.
	\newblock {\em Phys. Rev. A}, 90(3), Sep 2014.
	
	\bibitem{Brun2006}
	T.~Brun, I.~Devetak, and M.-H. Hsieh.
	\newblock Correcting quantum errors with entanglement.
	\newblock {\em Science}, 314(5798):436--439, Oct 2006.
	
	\bibitem{Calderbank1998}
	A.~R. Calderbank, E.~M. Rains, P.~W. Shor, and N.~J.~A. Sloane.
	\newblock Quantum error correction via codes over {GF}(4).
	\newblock {\em {IEEE} Trans. Inform. Theory}, 44(4):1369--1387, Jul 1998.
	
	\bibitem{Campbell2017}
	E.~T. Campbell, B.~M. Terhal, and C.~Vuillot.
	\newblock Roads towards fault-tolerant universal quantum computation.
	\newblock {\em Nature}, 549(7671):172--179, Sep 2017.
	
	\bibitem{Chen2015}
	B.~Chen, S.~Ling, and G.~Zhang.
	\newblock Application of constacyclic codes to quantum {MDS} codes.
	\newblock {\em {IEEE} Trans. Inform. Theory}, 61(3):1474--1484, Mar 2015.
	
	\bibitem{Chen2005}
	H.~Chen, S.~Ling, and C.~Xing.
	\newblock Quantum codes from concatenated algebraic-geometric codes.
	\newblock {\em {IEEE} Trans. Inform. Theory}, 51(8):2915--2920, Aug 2005.
	
	\bibitem{Cross2009}
	A.~Cross, G.~Smith, J.~A. Smolin, and B.~Zeng.
	\newblock Codeword stabilized quantum codes.
	\newblock {\em {IEEE} Trans. Inform. Theory}, 55(1):433--438, Jan 2009.
	
	\bibitem{Edel}
	Yves Edel.
	\newblock {Parameters of some $GF(3)$-linear quantum twisted codes}.
	\newblock Online available at
	\url{https://www.mathi.uni-heidelberg.de/~yves/Matritzen/QTBCH/QTBCHTab3.html},
	2000.
	\newblock Accessed on 2019-01-17.
	
	\bibitem{EJKL13}
	M.~F. Ezerman, S.~Jitman, H.~M. Kiah, and S.~Ling.
	\newblock Pure asymmetric quantum {MDS} codes from {CSS} construction: {A}
	complete characterization.
	\newblock {\em Int. J. Quantum Inf.}, 11(03):1350027, Apr 2013.
	
	\bibitem{Ezerman2013}
	M.~F. Ezerman, S.~Jitman, S.~Ling, and D.~V. Pasechnik.
	\newblock {CSS}-like constructions of asymmetric quantum codes.
	\newblock {\em {IEEE} Trans. Inform. Theory}, 59(10):6732--6754, Oct 2013.
	
	\bibitem{Ezerman2019}
	M.~F. Ezerman, S.~Ling, B.~\"{O}zkaya, and P.~Sol\'{e}.
	\newblock Good stabilizer codes from quasi-cyclic codes over $\mathbb{F}_4$ and
	$\mathbb{F}_9$.
	\newblock In {\em Proc. Int. Symp. Inform. Theory ({ISIT})}. {IEEE}, Jul 2019.
	
	\bibitem{Feng2006}
	K.~Feng, S.~Ling, and C.~Xing.
	\newblock Asymptotic bounds on quantum codes from algebraic geometry codes.
	\newblock {\em {IEEE} Trans. Inform. Theory}, 52(3):986--991, Mar 2006.
	
	\bibitem{Feng2004}
	K.~Feng and Z.~Ma.
	\newblock A finite {Gilbert{\textendash}Varshamov} bound for pure stabilizer
	quantum codes.
	\newblock {\em {IEEE} Trans. Inform. Theory}, 50(12):3323--3325, Dec 2004.
	
	\bibitem{Feynman1982}
	R.~P. Feynman.
	\newblock Simulating physics with computers.
	\newblock {\em Int. J. Theor. Phys.}, 21(6-7):467--488, Jun 1982.
	
	\bibitem{Forney2007}
	G.~D. Forney, M.~Grassl, and S.~Guha.
	\newblock Convolutional and tail-biting quantum error-correcting codes.
	\newblock {\em {IEEE} Trans. Inform. Theory}, 53(3):865--880, Mar 2007.
	
	\bibitem{GGHR2018}
	C.~Galindo, O.~Geil, F.~Hernando, and D.~Ruano.
	\newblock Improved constructions of nested code pairs.
	\newblock {\em {IEEE} Trans. Inform. Theory}, 64(4):2444--2459, Apr 2018.
	
	\bibitem{Galindo2018}
	C.~Galindo, F.~Hernando, and R.~Matsumoto.
	\newblock Quasi-cyclic constructions of quantum codes.
	\newblock {\em Finite Fields Appl.}, 52:261--280, Jul 2018.
	
	\bibitem{Gottesman97}
	D.~E. Gottesman.
	\newblock {\em Stabilizer Codes and Quantum Error Correction}.
	\newblock PhD thesis, California Institute of Technology, 1997.
	
	\bibitem{Grassl}
	M.~Grassl.
	\newblock {Bounds on the minimum distance of linear codes and quantum codes}.
	\newblock Online available at \url{http://www.codetables.de}, 2007.
	\newblock Accessed on 2020-02-20.
	
	\bibitem{Grassl2011}
	M.~Grassl.
	\newblock Variations on encoding circuits for stabilizer quantum codes.
	\newblock In Chee~Y.M. et~al., editor, {\em Coding and Cryptology. IWCC 2011.
		Lecture Notes in Comput. Sci.}, volume 6639, pages 142--158. Springer Berlin
	Heidelberg, 2011.
	
	\bibitem{GBR04}
	M.~Grassl, T.~Beth, and M.~R\"{o}tteler.
	\newblock On optimal quantum codes.
	\newblock {\em Int. J. Quantum Inf.}, 02(01):55--64, 2004.
	
	\bibitem{Grassl2015}
	M.~Grassl and M.~R\"{o}tteler.
	\newblock Quantum {MDS} codes over small fields.
	\newblock In {\em Proc. {IEEE} Int. Symp. Inf. Theory ({ISIT})}. {IEEE}, Jun
	2015.
	
	\bibitem{Guardia2013}
	G.~G.~La Guardia.
	\newblock Asymmetric quantum codes: new codes from old.
	\newblock {\em Quantum Inf. Process.}, 12(8):2771--2790, Mar 2013.
	
	\bibitem{Guenda2019}
	K.~Guenda, T.~A. Gulliver, S.~Jitman, and S.~Thipworawimon.
	\newblock Linear $\ell$-intersection pairs of codes and their applications.
	\newblock {\em Des. Codes and Cryptogr.}, 88(1):133--152, Jan 2020.
	
	\bibitem{Ioffe2007}
	L.~Ioffe and M.~M{\'{e}}zard.
	\newblock Asymmetric quantum error-correcting codes.
	\newblock {\em Phys. Rev. A}, 75(3), Mar 2007.
	
	\bibitem{Ketkar2006}
	A.~Ketkar, A.~Klappenecker, S.~Kumar, and P.~K. Sarvepalli.
	\newblock Nonbinary stabilizer codes over finite fields.
	\newblock {\em {IEEE} Trans. Inform. Theory}, 52(11):4892--4914, Nov 2006.
	
	\bibitem{Kitaev2003}
	A.~Yu. Kitaev.
	\newblock Fault-tolerant quantum computation by anyons.
	\newblock {\em Ann. Phys.}, 303(1):2--30, Jan 2003.
	
	\bibitem{Knill1997}
	E.~Knill and R.~Laflamme.
	\newblock Theory of quantum error-correcting codes.
	\newblock {\em Phys. Rev. A}, 55(2):900--911, Feb 1997.
	
	\bibitem{Knill1998}
	E.~Knill, R.~Laflamme, and W.~H. Zurek.
	\newblock Resilient quantum computation.
	\newblock {\em Science}, 279(5349):342--345, 1998.
	
	\bibitem{Laflamme1996}
	R.~Laflamme, C.~Miquel, J.~P. Paz, and W.~H. Zurek.
	\newblock Perfect quantum error correcting code.
	\newblock {\em Phys. Rev. Lett.}, 77(1):198--201, Jul 1996.
	
	\bibitem{Lidar2013}
	D.~A. Lidar and T.~A. Brun, editors.
	\newblock {\em Quantum Error Correction}.
	\newblock Cambridge University Press, 2013.
	
	\bibitem{Lidl1996}
	R.~Lidl and H.~Niederreiter.
	\newblock {\em Finite Fields}.
	\newblock Cambridge University Press, Oct 1996.
	
	\bibitem{Ling2010}
	S.~Ling, J.~Luo, and C.~Xing.
	\newblock Generalization of {S}teane's enlargement construction of quantum
	codes and applications.
	\newblock {\em {IEEE} Trans. Inform. Theory}, 56(8):4080--4084, Aug 2010.
	
	\bibitem{Lisonek2014}
	P.~Lison{\v{e}}k and V.~Singh.
	\newblock Quantum codes from nearly self-orthogonal quaternary linear codes.
	\newblock {\em Des. Codes and Cryptogr.}, 73(2):417--424, Feb 2014.
	
	\bibitem{Macwilliams1963}
	F.~J. MacWilliams.
	\newblock A theorem on the distribution of weights in a systematic code.
	\newblock {\em Bell System Tech. J.}, 42(1):79--94, Jan 1963.
	
	\bibitem{Matsumoto2017}
	R.~Matsumoto.
	\newblock Two {Gilbert{\textendash}Varshamov}-type existential bounds for
	asymmetric quantum error-correcting codes.
	\newblock {\em Quantum Inf. Process.}, 16(12), Oct 2017.
	
	\bibitem{Preskill1997}
	J.~Preskill.
	\newblock Fault-tolerant quantum computation, 1997.
	
	\bibitem{Rains99}
	E.~M. Rains.
	\newblock Nonbinary quantum codes.
	\newblock {\em {IEEE} Trans. Inform. Theory}, 45(6):1827--1832, Sep 1999.
	
	\bibitem{Sarvepalli2009}
	P.~K. Sarvepalli, A.~Klappenecker, and M.~R\"{o}tteler.
	\newblock Asymmetric quantum codes: {C}onstructions, bounds and performance.
	\newblock {\em Proc. Roy. Soc. London Ser. A}, 465(2105):1645--1672, Mar 2009.
	
	\bibitem{Schumacher1995}
	B.~Schumacher.
	\newblock Quantum coding.
	\newblock {\em Phys. Rev. A}, 51(4):2738--2747, Apr 1995.
	
	\bibitem{Shor1995}
	P.~W. Shor.
	\newblock Scheme for reducing decoherence in quantum computer memory.
	\newblock {\em Phys. Rev. A}, 52(4):R2493--R2496, Oct 1995.
	
	\bibitem{Steane1996}
	A.~M. Steane.
	\newblock Multiple-particle interference and quantum error correction.
	\newblock {\em Proc. Roy. Soc. London Ser. A}, 452(1954):2551--2577, Nov 1996.
	
	\bibitem{Steane1999}
	A.~M. Steane.
	\newblock Enlargement of {Calderbank-Shor-Steane} quantum codes.
	\newblock {\em {IEEE} Trans. Inform. Theory}, 45(7):2492--2495, 1999.
	
	\bibitem{Steane2003}
	A.~M. Steane.
	\newblock Overhead and noise threshold of fault-tolerant quantum error
	correction.
	\newblock {\em Phys. Rev. A}, 68(4), Oct 2003.
	
	\bibitem{Vardy1997}
	A.~Vardy.
	\newblock The intractability of computing the minimum distance of a code.
	\newblock {\em {IEEE} Trans. Inform. Theory}, 43(6):1757--1766, 1997.
	
	\bibitem{Wang2010}
	L.~Wang, K.~Feng, S.~Ling, and C.~Xing.
	\newblock Asymmetric quantum codes: Characterization and constructions.
	\newblock {\em {IEEE} Trans. Inform. Theory}, 56(6):2938--2945, Jun 2010.
	
	\bibitem{Wilde2008}
	M.~M. Wilde and T.~A. Brun.
	\newblock Optimal entanglement formulas for entanglement-assisted quantum
	coding.
	\newblock {\em Phys. Rev. A}, 77(6), Jun 2008.
	
	\bibitem{Wootters1982}
	W.~K. Wootters and W.~H. Zurek.
	\newblock A single quantum cannot be cloned.
	\newblock {\em Nature}, 299(5886):802--803, Oct 1982.
	
\end{thebibliography}
\end{document}